\pdfoutput=1

\documentclass[11pt]{article}
\usepackage[noblocks]{authblk}
\usepackage{titlesec}
\usepackage{geometry}
\usepackage{multirow}
\usepackage{fancyhdr}
\usepackage[numbers]{natbib}
\usepackage{graphicx}
\usepackage[english]{babel}
\usepackage{amsmath}
\usepackage{amssymb}
\usepackage{amsfonts}
\usepackage{amsthm}
\usepackage{thmtools} 
\usepackage{thm-restate}
\usepackage{accents}
\usepackage[noend]{algpseudocode}
\usepackage{caption}
\usepackage{subcaption}

\usepackage[dvipsnames]{xcolor}
\definecolor{cb-black}      {RGB}{  0,   0,   0}
\definecolor{cb-blue-green} {RGB}{  0,  073,  073}
\definecolor{cb-green-sea}  {RGB}{  0, 146, 146}
\definecolor{cb-rose}       {RGB}{255, 109, 182}
\definecolor{cb-salmon-pink}{RGB}{255, 182, 119}
\definecolor{cb-purple}     {RGB}{ 73,   0, 146}
\definecolor{cb-blue}       {RGB}{ 0, 109, 219}
\definecolor{cb-lilac}      {RGB}{182, 109, 255}
\definecolor{cb-blue-sky}   {RGB}{109, 182, 255}
\definecolor{cb-blue-light} {RGB}{182, 219, 255}
\definecolor{cb-burgundy}   {RGB}{146,   0,   0}
\definecolor{cb-brown}      {RGB}{146,  73,   0}
\definecolor{cb-clay}       {RGB}{219, 209,   0}
\definecolor{cb-green-lime} {RGB}{ 36, 255,  36}
\definecolor{cb-yellow}     {RGB}{255, 255, 109}

\usepackage[colorlinks,linkcolor=cb-burgundy,urlcolor=cb-burgundy,citecolor=cb-burgundy,destlabel]{hyperref}
\usepackage[nameinlink,capitalise]{cleveref}
\usepackage{nicefrac}
\usepackage{stmaryrd}
\usepackage[normalem]{ulem}
\usepackage{lipsum}
\usepackage{enumitem}

\makeatletter
\let\OldStatex\Statex
\renewcommand{\Statex}[1][3]{
  \setlength\@tempdima{\algorithmicindent}
  \OldStatex\hskip\dimexpr#1\@tempdima\relax}
\makeatother

\algnewcommand{\Inp}{\textbf{Input:}\space}
\algnewcommand{\Out}{\textbf{Output:}\space}
\newcommand{\Input}{\Statex[-1] \Inp }
\newcommand{\Output}{\Statex[-1] \Out }
\newcommand{\Blank}{\Statex[-1]}

\makeatletter
\providecommand\theHALG@line{\thealgorithm.\arabic{ALG@line}}
\makeatother

\newtheorem{theorem}{Theorem}
\counterwithin{figure}{section}
\counterwithin{table}{section}
\counterwithin{theorem}{section}
\newtheorem{proposition}[theorem]{Proposition}

\newtheorem{definition}[theorem]{Definition}
\newtheorem{problem}[theorem]{Problem}

\theoremstyle{remark}

\declaretheoremstyle[
  notefont=\mdseries, notebraces={(}{)},
  bodyfont=\normalfont,
  postheadspace=0em,
  headpunct=
]{algostyle}
\theoremstyle{algostyle}

\newtheorem{algorithm}[theorem]{Algorithm}

\newcommand{\ket}[1]{|#1\rangle}

\newcommand*{\f}{\mathbb{F}}

\newcommand*{\ip}[1]{ \langle #1 \rangle }

\newcommand{\comm}[1]{ \llbracket #1 \rrbracket }

\newcommand*{\nM}{n_M}

\newcommand*{\Co}{R} 

\newcommand*{\Vph}{V_{\varphi}} 
\newcommand*{\Uh}{U_{H}} 
\newcommand*{\Uph}{U_{\varphi}} 
\newcommand*{\ph}{\zeta_{8}}

\newcommand{\algemph}[1]{\colorbox{cb-blue-light!35!white}{#1}}

\newcommand{\emphphase}[1]{\colorbox{cb-purple!10!white}{#1}}

\makeatletter
\def\smalloverbrace#1{\mathop{\vbox{\m@th\ialign{##\crcr\noalign{\kern3\p@}
  \tiny\downbracefill\crcr\noalign{\kern3\p@\nointerlineskip}
  $\hfil\displaystyle{#1}\hfil$\crcr}}}\limits}
\makeatother

\DeclareEmphSequence{\bfseries,\mdseries}

\title{
Phased outcome-complete simulation
}

\author[1]{Vadym Kliuchnikov\thanks{Current address: NVIDIA, Toronto, ON M5V 1K4, Canada}}
\affil[1]{Microsoft Quantum, Toronto, ON M5J 0E7, Canada}
\author[2]{Adam Paetznick}
\author[2]{Marcus P. da Silva}
\affil[2]{Microsoft Quantum, Redmond, WA 98052, USA}

\begin{document}

\maketitle

\begin{abstract}
We generalize the polynomial-time outcome-complete simulation algorithm for stabilizer circuits~[Kliuchnikov, Beverland, and Paetznick,~\hyperlink{https://arxiv.org/abs/2309.08676}{arXiv:2309.08676}] to track global phases exactly, yielding what we call \emph{phased outcome-complete simulation}.
The original algorithm enabled equivalence checking of stabilizer circuits with intermediate measurements and conditional Pauli corrections for all input states and all measurement outcomes simultaneously, but it tracked quantum states only up to a global phase.
Our generalization removes this limitation and enables equivalence checking for an important family of non-stabilizer circuits: stabilizer circuits augmented with single-qubit rotations $\exp(i\alpha Z)$ by symbolic angles.
Two such circuits are equivalent if they implement the same quantum channel for all values of the symbolic angles and all measurement outcomes, given a one-to-one correspondence between rotation angles in the two circuits and a mapping between measurement outcomes.
This model enables testing of compilation algorithms that transform the Clifford portions of a computation while preserving rotation angles. Examples include Pauli-based computation, edge-disjoint path compilation for surface codes, and custom compilation strategies for reversible circuits such as adders, multipliers, and table lookups.
Our efficient classical verification methods extend naturally to circuits with outcome-parity-conditional Pauli gates and intermediate measurements, features that are ubiquitous in fault-tolerant quantum computing but are rarely addressed by existing equivalence-checking approaches.
\end{abstract}

\section{Introduction}

Verification and debugging of quantum circuits are essential components of the quantum computing development cycle.
As quantum computers scale toward fault-tolerant architectures, the circuits describing quantum algorithms undergo increasingly sophisticated compilation steps.
Ensuring correctness at each stage of this compilation pipeline is critical; however, verifying equivalence of general quantum circuits is QMA-complete~\cite{janzing2005}.
We therefore seek efficient equivalence-checking methods for important but restricted families of quantum circuits.

Previous work~\cite{KBP, MicrosoftQDKEC} introduced an \emph{outcome-complete simulation} algorithm that enables equivalence checking of stabilizer circuits for all possible input states and measurement outcomes.
This class of circuits includes those with intermediate measurements, random outcomes, and Pauli unitaries conditioned on parities of measurement results.

It is standard to reduce equivalence checking of two circuits for all possible input states and measurement outcomes to equivalence checking of two quantum states for all possible measurement outcomes by considering the Choi states of the circuits.
The key insight was that states prepared by stabilizer circuits can be reduced to a canonical \emph{general form} that captures the circuit's measurement outcomes and the prepared quantum state as functions of random outcomes in a single representation.
The algorithm for computing such a representation is called outcome-complete simulation.\footnote{Throughout, we use \emph{simulation} to mean the computation of a representation that captures the state prepared by a stabilizer circuit, the circuit's measurement outcomes, and their dependence on random outcomes. This is a departure from the standard notions of weak and strong simulation of quantum circuits.}
However, this algorithm tracked quantum states only up to a global phase, which limits its applicability beyond the stabilizer formalism.

In this work, we generalize outcome-complete simulation to track global phases exactly, yielding what we call \textbf{phased outcome-complete simulation}.
This seemingly modest extension has significant practical consequences: it enables equivalence checking for certain families of non-stabilizer circuits.

Specifically, we consider stabilizer circuits augmented with single-qubit rotation gates $\exp(i\alpha Z)$ for \emph{symbolic} angles $\alpha$.
Two such circuits are equivalent if they implement the same quantum channel for all values of the symbolic angles and all measurement outcomes, given a correspondence between angles and a mapping between outcomes.
We reduce this \textbf{non-stabilizer} equivalence problem to checking whether certain families of \textbf{stabilizer} states are exactly equal for all measurement outcomes --- not merely equal up to a global phase.
This is where phased outcome-complete \textbf{stabilizer} simulation becomes essential.
This is a restricted setting, as two circuits might still be equivalent for specific values of symbolic angles, but our algorithm cannot detect such cases.
Additionally, some equivalent circuits may have different numbers of single-qubit rotations, and our algorithm does not apply to them.

However, the symbolic-angle model captures an important class of compilation algorithms in fault-tolerant quantum computing.
Modern compilation strategies for surface codes and quantum LDPC codes often transform the Clifford portions of a computation while relocating non-Clifford gates without modifying their rotation angles.
Examples include Pauli-based computation~\cite{beverland2022, aasen2025topologicallyfaulttolerantquantumcomputer}, edge-disjoint path compilation~\cite{BeverlandKliuchnikovSchoute}, and custom compilation strategies for reversible circuits such as table lookups, adders, and multipliers~\cite{haner2022spacetime, gidney2019flexiblelayoutsurfacecode}.
These algorithms operate symbolically on rotation angles and do not require their numerical values during compilation.
Our verification framework is well suited to testing the correctness of such compilers.

A notable feature of our setting is the presence of outcome-parity-conditional Pauli gates and intermediate measurements.
Most prior equivalence-checking work focuses on unitary gate sequences~\cite{Amy2019, PehamBurgholzerWille, burgholzer2020}, yet practical fault-tolerant circuits routinely include measurements with random outcomes and gates conditioned on those outcomes.
Indeed, some of the leading approaches for logical operations on encoded qubits are entirely measurement-based~\cite{Horsman2012, Beverland2022Disjoint, haner2022spacetime}.
Our methods handle these features naturally, extending the applicability of efficient classical verification to a broader and more realistic class of circuits.

\paragraph{Paper outline.}
\Cref{sec:verification} presents the verification problem and shows how non-stabilizer equivalence reduces to exact equality of stabilizer states.
\Cref{sec:phased-clifford-unitaries,sec:measurement-as-unitary,sec:aux-separation} develop the technical components: an efficient representation for phased Clifford unitaries, a phase-aware measurement lemma, and auxiliary qubit separation.
The main algorithm appears in \cref{alg:outcome-complete-stab-phase-sim}.
\Cref{app:css-orbits,tab:bruhat-decompositions} contain supporting algorithms and lookup tables.

\section{Preliminaries and notation}

We closely follow notation and terminology introduced in Section~1 in~\cite{KBP}.
We find it useful to refer to Hermitian Pauli unitaries $\pm \{I,X,Y,Z\}^{\otimes n}$, as \emph{Pauli observables}. 
For bit-vector $x$ of length $n$, we define $X^x = X^{x_1} \otimes \ldots \otimes X^{x_n}$.
In our algorithms it is convenient to represent multi-qubit Pauli unitaries $P$ as $i^{s(P)} X^{x(P)} Z^{z(P)}$, where $x(P)$ denotes the $x$ bits of $P$, $z(P)$ denotes the $z$ bits of $P$, and $s(P)$ is the $xz$-phase of $P$.

We define the \emph{controlled-Pauli unitary} for any distinct commuting Pauli observables $P_1,P_2$ as: 
\begin{equation}
\label{eq:controlled-pauli}
\Lambda(P_1,P_2)
=
\frac{I+P_1}{2} +  \frac{I-P_1}{2} \cdot P_2.
\end{equation}
For some useful identities involving $\Lambda(P_1,P_2)$ and Pauli exponents $e^{i\frac{\pi}{4}P}$ see Section~2 in~\cite{KBP}. The following identity will be used later:
\begin{align}
\label{eq:controlled-pauli-image}
\Lambda(P_1,P_2)~Q~\Lambda(P_1,P_2) & = P_1^{\comm{P_2,Q}}~Q~P_2^{\comm{P_1,Q}}.
\end{align}

We follow the following definition of stabilizer circuits from~\cite{KBP}:
\begin{definition}[Stabilizer circuit]
\label{def:stabilizer-circuit}
A \emph{stabilizer circuit} is any sequence of the following elementary operations,
that we call \emph{stabilizer operations},
\begin{itemize}[noitemsep]
    \item allocations of qubits initialized to zero states,
    \item allocation of classical random bits distributed as fair coins, 
    \item Clifford and Pauli unitaries,
    \item non-destructive Pauli measurements,
    \item deallocation of qubits in zero states,
    \item Pauli unitaries conditioned on the parity of sets of measurement outcomes and classical random bits from earlier in the sequence.
\end{itemize}
Any stabilizer circuit starts with qubits in an arbitrary state that we call \emph{input qubits}. 
Qubits that remain after executing the circuit are \emph{output qubits}.
We call the sequence of all measurement outcomes and classical random bits allocated by the circuit the \emph{circuit outcome vector}.
\end{definition}

We need to distinguish Clifford unitaries specified up to a global phase (that is, as elements of the projective unitary group) and Clifford unitaries specified fully as a unitary matrix. We refer to the former simply as Clifford unitaries, as is customary, and to the latter as \textbf{phased} Clifford unitaries. 

We say that a Clifford unitary $C$ is a \textbf{CSS Clifford} if the images $C X_j C^\dagger$ are tensor products of $X$s and the images $C Z_j C^\dagger$ are tensor products of $Z$s. 
These Clifford unitaries can be described by an invertible square matrix $A$ as $U_A|r\rangle = |Ar\rangle$.
Conjugating an $n$-qubit CSS Clifford by the $n$-th tensor power of Hadamard gives the dual CSS Clifford (this is a corollary of Proposition~A.2 in~\cite{KBP}):
\begin{equation} \label{eq:css-dual}
    (H^{\otimes n}) U_A (H^{\otimes n}) = U_{A^{-T}}, \text{ where } A^{-T} = {(A^T)^{-1}}
\end{equation}
For example, operators $\Lambda(X^x,Z^z)$ with $\ip{x,z}=0$ are CSS Clifford~(\cref{eq:controlled-pauli-image}). We call sequences of such operators \textbf{CSS circuits}, the \textbf{length} of the circuit is the length of the sequence. CSS Clifford unitaries are products of controlled-$X$s~\cite{AaronsonGottesman2004}. Similarly, we call Clifford unitaries \textbf{phase-CSS} if they are products of $S$, controlled-$Z$, and controlled-$X$. This class of unitaries plays an important role in the Bruhat decomposition~\cite{MaslovRoetteler2018} of Clifford unitaries and has a well-understood structure of its binary-symplectic representation~\cite{MaslovRoetteler2018}. 

Matrices other than Clifford unitaries are over $\mathbb{F}_2$, vectors are column-vectors over $\mathbb{F}_2$,
and the inner product $\ip{a,r}$ of vectors $a,r$ is also over $\mathbb{F}_2$, as are matrix additions and multiplications. We use the notation $n(C)$ for the number of qubits on which Clifford unitary $C$ is specified, and $n(r)$
for the dimension of vector $r$. Vector coordinates are $1$-indexed; for example, $r_{n(r)}$ is the last coordinate of $r$. 
We use $e_j$ to denote the standard basis vector with a $1$ in the $j$-th coordinate and $0$s elsewhere, 
and $[n]$ for the set $\{1,2,\ldots,n\}$.

\section{Circuit verification via phased outcome-complete simulation}
\label{sec:verification}

Our goal is to introduce algorithmic tools for circuit verification beyond stabilizer circuits considered in~\cite{KBP}. 
We start with a simple motivating example. 
Let $C_1, C_2, D_1, D_2$ be Clifford unitaries. 
We would like to check whether the following is true:
\begin{equation} \label{eq:non-stab-equality}
     \text{For all } \alpha \in \mathbb{R},\,  C_1 \exp(i \alpha Z_1 ) C_2 |0\rangle \stackrel{?}{=} D_1 \exp(i \alpha Z_1 ) D_2 |0\rangle.
\end{equation}

We can check that the above holds if and only if the following stabilizer states are equal
\begin{equation} \label{eq:stab-equality}
\text{For all } a \in \{0,1\},\, C_1 Z_1^a C_2 |0\rangle \stackrel{?}{=} D_1 Z_1^a D_2 |0\rangle. 
\end{equation}
For the argument, it is crucial that the equality is exact rather than up to a global phase.
Recall that $\exp(i \alpha Z_1 ) = \cos(\alpha) I + i \sin(\alpha) Z_1$. 
By choosing $\alpha = 0, \nicefrac{\pi}{2}$, \cref{eq:stab-equality} immediately follows from \cref{eq:non-stab-equality}. Similarly, \cref{eq:non-stab-equality} follows from \cref{eq:stab-equality} by taking linear combinations of \cref{eq:stab-equality} evaluated at $a=0,1$ and weights $\cos(\alpha), i\sin(\alpha)$.

\cref{eq:stab-equality} can be interpreted as equality of two stabilizer state preparation circuits with a random bit $a$.
The solution to checking equality of even more general stabilizer state preparation circuits with random bits, stabilizer measurements, and Pauli unitaries conditional on parities of measurement outcomes relies on 
the outcome-complete stabilizer simulation algorithm introduced in~\cite{KBP} that solves the following problem: 

\begin{problem}[Outcome-complete stabilizer circuit simulation]
\label{def:outcome-complete-stab-sim}
Consider any stabilizer circuit with no input qubits, $n$ output qubits, and a length-$\nM$ outcome vector. 
Find the vector of non-zero conditional probabilities $\overrightarrow{q} \in \{1,1/2\}^{\nM}$,
a Clifford unitary $\Co$,
matrices $A$ and $M$ and a vector $v_0$ with entries in $\f_2$
that satisfy the following properties:
\begin{itemize}[noitemsep]
    \item each possible outcome vector $v$ is an element of the set $\{M r : r \in \{0,1\}^{n_r} \}$, where $n_r =|\{ \overrightarrow{q}_k =1/2 : k \in [\nM] \}|$, 
    \item for any outcome vector, $\overrightarrow{q}_j$ is the probability of obtaining outcome $v_j$ given previous outcomes $v_1,\ldots,v_{j-1}$,
    \item $\Co | A r \rangle $ equals to the output state of the circuit given the outcome vector $v = v_0 + M r$ up to \textbf{a global phase} .
\end{itemize} 
\end{problem}

The only reason we cannot readily apply the outcome-complete simulation algorithm to the above problem 
is that it only considers equality up to a global phase, and in our example we need exact equality.
In this work, we generalize the outcome-complete simulation problem to include global phase~(\cref{def:outcome-complete-phase-stab-sim}) and provide an efficient algorithm for it~(\cref{alg:outcome-complete-stab-phase-sim}).
Note that our example readily generalizes to many angles $\alpha$ in \cref{eq:non-stab-equality} and
to using stabilizer channels instead of the Clifford unitaries $C_j, D_j$. 
Similarly, we can include circuits with inputs by considering Choi states of the circuits.
Here we focus on outcome-complete simulation that includes global phase, which we call \textbf{phased-outcome-complete simulation}, as the core idea to enable the above-mentioned generalizations.  

\section{Phased outcome-complete simulation algorithm}

The generalization of outcome-complete simulation to include global phase is straightforward; however, we need to address a few technical details. The key aspects that need upgrading are:
\begin{enumerate}
    \item Clifford unitary operations, using an efficient data structure for phased Clifford unitaries (\cref{sec:phased-clifford-unitaries}).
    \item Conditional Pauli operators.
    \item Measurements with random outcomes, generalizing Proposition~2.2 of~\cite{KBP} (\cref{sec:measurement-as-unitary}).
    \item Separating auxiliary from output qubits, generalizing Appendix~C of~\cite{KBP} (\cref{sec:aux-separation}).
\end{enumerate}

We next proceed to formally state the phased outcome-complete stabilizer circuit simulation~\cref{def:outcome-complete-phase-stab-sim}, 
provide an efficient algorithm~(\cref{alg:outcome-complete-stab-phase-sim}) for it, and sketch a correctness proof. 
The formal statement of the problem is as follows. We \emphphase{highlight} the differences from the outcome-complete simulation problem~\cref{def:outcome-complete-stab-sim}.

\begin{problem}[Phased outcome-complete stabilizer circuit simulation]
\label{def:outcome-complete-phase-stab-sim}
Consider any stabilizer circuit with no input qubits, $n$ output qubits, and a length-$\nM$ outcome vector. 
Find the vector of non-zero conditional probabilities $\overrightarrow{q} \in \{1,1/2\}^{\nM}$,
a \emphphase{phased} Clifford unitary $\Co$,
matrices $A$, \emphphase{$B$}, and $M$, and vectors $v_0$, \emphphase{$p, s$} with entries in $\f_2$
that satisfy the following properties:
\begin{itemize}[noitemsep]
    \item each possible outcome vector $v$ is an element of the set $\{v_0 + M r : r \in \{0,1\}^{n_r} \}$, where $n_r =|\{ \overrightarrow{q}_k =1/2 : k \in [\nM] \}|$, 
    \item for any outcome vector, $\overrightarrow{q}_j$ is the probability of obtaining outcome $v_j$ given previous outcomes $v_1,\ldots,v_{j-1}$,
    \item \emphphase{$i^{\ip{p,r}} (-1)^{\ip{Br+s,r}}$}$\Co | A r \rangle$  is the output state of the circuit given the outcome vector $v = v_0 + M r$.
\end{itemize} 
\end{problem}

The global phase dependence of the simulation output state on the random outcomes is similar to phase expressions for stabilizer states. The power of $i$ includes a linear function of the random outcome vector $r$, and the
power of $-1$ is quadratic in the random outcome vector $r$. Similar to how the difference between outcome-complete and outcome-specific simulation is highlighted in~\cite{KBP}, 
we \emphphase{highlight} differences related to tracking global phase in~\cref{alg:outcome-complete-stab-phase-sim}.
We also keep \algemph{highlights} that distinguish outcome-specific and outcome-complete simulation~\cite{KBP}.

\begin{figure*}[p]
\begin{algorithm}[\texttt{Phased outcome-complete stabilizer circuit simulation}] \label{alg:outcome-complete-stab-phase-sim}
\begin{algorithmic}[1]
\Blank
\Input A stabilizer circuit $\mathcal{C}$ with no input qubits, $n$ output qubits, and $\nM$ outcomes. 
\Output 
\begin{itemize}[noitemsep,topsep=0pt]
    \item a vector $\overrightarrow{q} \in \{1,1/2\}^{\nM}$ of conditional probabilities,
    \item a \emphphase{phased} Clifford unitary $\Co$ 
    \item matrices $A, M, B$ and column-vectors $v_0,p,s$ with entries in $\f_2$
\end{itemize}
that satisfy conditions of~\cref{def:outcome-complete-phase-stab-sim}, additionally $M^T$ is in reduced row echelon form
and the entries of $v_0$ corresponding to the row rank profile of $M$ are zero.

\State Initialize an empty vector $\overrightarrow{q}$, zero-qubit Clifford unitary \emphphase{with phase} $\Co$, dimension-zero matrices \algemph{$A$}, \emphphase{$M, B$}, column-vectors \algemph{$v_0$}, \emphphase{$p ,s$}.

\For{$g$ in $\mathcal{C}$}

\hrulefill\Comment{allocation}
\If{$g$ allocates qubit $j$,\label{line:phase-complete-allocate}}
  \State replace $\Co \leftarrow \Co \otimes_j I_2$,
  \State \algemph{insert a zero row into $A$ after row $j-1$.}
\ElsIf{$g$ allocates a random bit\label{line:phase-complete-random}}
  \State append $1/2$ to $\overrightarrow{q}$,
  \State \algemph{add a zero column to $A$}, \emphphase{add a zero row and column to $B$}
  \State \algemph{append zero to the end of $v_0$}, \emphphase{$p$, $s$},
  \State \algemph{add a zero column and row to $M$ and set the last bit in the row to $1$.}

\hrulefill\Comment{unitaries}
\ElsIf{$g$ is a unitary $U$\label{line:phase-complete-unitary}} 
  \State replace $\Co \leftarrow U \Co$.
\ElsIf{$g$ applies a Pauli unitary $P$ if $\langle c\rangle = c_0$, \\$\quad\quad\quad\quad\quad$ where $\langle c\rangle$ is the parity of outcomes indicated by $c\in \mathbb{F}_2^{\nM}$,\label{line:phase-complete-conditional}}
  \State replace $\Co \leftarrow P^{c_0 + 1} \Co$ \algemph{followed by $\Co \leftarrow P^{\ip{v_0,c}} \Co$ followed by $a \leftarrow M^T c$}
  \State \algemph{find preimage $i^l X^x Z^z=\Co^\dagger P\Co$ } \label{line:phase-complete-conditional-r} \Comment{apply $P$ conditioned on bits of $r$ indicated by $a$}
  \State \emphphase{replace $B \rightarrow B + a z^T A $}, \algemph{replace $A \rightarrow A + x a^T $}
  \State \emphphase{update $p,s$ to $p',s'$ that satisfy $i^{\ip{p',r}} (-1)^{\ip{s',r}} = i^{\ip{p,r} } (-1)^{\ip{s,r}} (i^l)^{\ip{a,r}}$ for all $r$}

\hrulefill\Comment{measurements}
\ElsIf{$g$ measures Pauli $P$ given a hint Pauli $P'$ such that $\comm{P,P'}=1$ \\$\quad\quad\quad\quad\quad$ and  $P'\Co|\textbf{0}\rangle = \pm \Co|\textbf{0}\rangle$\label{line:phase-complete-fast-measure}}
  \State find $b'$ and $\alpha$ such that preimage $ \Co^\dagger P'\Co = (-1)^{\alpha} Z^{b'}$,
  \State replace $\Co\leftarrow (-1)^\alpha e^{i\frac{\pi}{4}(iP'P)} \Co$, followed by \algemph{$a \leftarrow A^T b'$}
  \State allocate a random bit (\cref{line:phase-complete-random}), \emphphase{replace $B \leftarrow B + (a\oplus 0)^T (a\oplus 1)$, $s_{n(s)} \leftarrow s_{n(s)} + \alpha$}
  \State \algemph{apply $P'$ conditioned on bits of $r$ indicated by $a \oplus 1$, according to \cref{line:phase-complete-conditional-r}.}\label{line:phase-complete-conditional-in-m}
\ElsIf{$g$ measures Pauli $P$,}
\State find preimage $Q = \Co^\dagger P \Co$. 
\If{$x(Q)=0$ (deterministic measurement)\label{line:phase-complete-deterministic-measure}} 
  \State append $1$ to $\overrightarrow{q}$,
  \State \algemph{add row $A^T z(Q)$ to $M$ and append $s(Q)/2$ to the end of $v_0$. }
\Else{ (uniformly random measurement)\label{line:phase-complete-random-measure}}
  \State let $j$ be the position of first non-zero bit of $x(Q)$, 
  \State find image $P' = \Co Z_j \Co^\dagger$ (which anticommutes with $P$).\label{line:phase-complete-image}
  \State measure $P$ with assertion $P'$ according to \cref{line:phase-complete-fast-measure}.
\EndIf
\EndIf
\EndFor
\State \Return vector $\overrightarrow{q}$, \emphphase{phased} Clifford unitary $\Co$,  matrices \algemph{$A,M$}, \emphphase{$B$}, vectors \algemph{$v_0$}, \emphphase{$p ,s$}.
\end{algorithmic}
\end{algorithm}
\end{figure*}

Next, we sketch the correctness proof for~\cref{alg:outcome-complete-stab-phase-sim} by examining the allocation, unitary, and measurement parts of the simulation algorithm. In the sketch, we follow the notation in the algorithm pseudocode.

Updates related to allocation of new qubits~(\cref{line:phase-complete-allocate}) ensure that the dimensions of $\Co,A$ match and that $\Co|Ar\rangle$ represents a state with an additional qubit in the zero state. Updates related to random bit allocation~(\cref{line:phase-complete-random}) ensure that the dimensions of 
vectors $v_0,p,s$ and matrices $A,B,M$ match the dimension of the random outcomes vector $r$ and that the functions $Ar$, $\ip{p',r}$,
$\ip{Br+s,r}$ evaluate to the same values as before.

Simulation of Clifford unitaries~(\cref{line:phase-complete-unitary}) differs only in using a new data structure to represent phased Clifford unitaries instead of the standard up-to-global-phase representation. The simulation of Pauli unitaries conditional on parities of measurement outcomes goes through the same process of converting dependence on measurement outcomes to dependence on random outcomes via matrix $M$ and shift $v_0$~(\cref{line:phase-complete-conditional}). Applying a Pauli operator $P$ conditional on random bits~(\cref{line:phase-complete-conditional-r}) involves a significant addition: we need to account for a conditional global phase 
when pulling a conditional Pauli through $\Co|A\rangle$: 
$$
 P^{\ip{r,a}} \Co|Ar\rangle =  \Co (i^l X^x Z^z)^{\ip{r,a}} |Ar\rangle = (i^l)^{\ip{r,a}} (-1)^{\ip{r,a}\ip{z,Ar}} \Co |x \ip{r,a} + Ar\rangle.
$$
The above equality translates into the rules for updating matrices $A,B$ and vectors $p,s$.

The deterministic measurement simulation~(\cref{line:phase-complete-deterministic-measure}) remains the same as in outcome-complete simulation because it does not affect the quantum state and only requires an update to the outcome map~($M, v_0$) relating random bits and circuit outcomes. The handling of random outcomes has been significantly updated~(\cref{line:phase-complete-random-measure}) due to changes in how we handle random measurements with hints~(\cref{line:phase-complete-fast-measure}). The main technical contribution is including the global phase when reducing uniformly random measurement to applying unitaries followed by conditional Pauli operations. This is detailed in~\cref{sec:measurement-as-unitary}. When the hint $P'$ for uniformly random measurement satisfies $ \Co^\dagger P'\Co = (-1)^{\alpha} Z^{b'}$, the quantum state $\Co|A r\rangle$ is stabilized by $(-1)^{\beta(r)} P'$, where $\beta(r) = \alpha + \ip{b',Ar}$. 
Following~\cref{prop:measure-as-exp}, the global phase needs to be updated by $(-1)^{\beta(\beta + \rho)}$,
where $\rho(r)$ is the value of the newly allocated random bit, 
$\Co$ must be replaced with $(-1)^\alpha e^{i\frac{\pi}{4}(iP'P)} \Co$, and Pauli $P'$ 
must be applied conditionally on $\rho(r) + \beta(r)$. 
Expressing functions $\rho(r) + \beta(r)$ in terms of random outcome vector $r$ gives us required expressions for updating $B,s$ and random outcome indicator $a \oplus 1$ for conditionally applying $P'$.

We have removed the explicit qubit deallocation step from the algorithm, as it is more efficient to reset qubits back to the zero state with a computational-basis measurement followed by a conditional Pauli $X$, as discussed in Appendix~C in~\cite{KBP}. The deallocated qubits can be reused when another qubit in the zero state needs to be allocated. This way, we only need to deallocate auxiliary qubits at the end of simulation. The details of handling the global phase are in \cref{sec:aux-separation}.

\subsection{Phased Clifford unitaries: an efficient representation}
\label{sec:phased-clifford-unitaries}

An efficient data structure for representing Clifford unitaries up to global phase is at the core of standard stabilizer simulation algorithms~\cite{AaronsonGottesman2004, KBP, Gidney2021}.
Using the Bruhat decomposition~\cite{MaslovRoetteler2018,BravyiMaslov2021} of Clifford unitaries and following the CH decomposition~\cite{BravyiBrowneAndCo} for representing stabilizer states with global phases, 
we show how to represent any phased Clifford unitary.

It follows from the Bruhat decomposition that any Clifford unitary can be written as 
$\Uph \Uh P \Vph$, where $\Uph$, $\Vph$ are phase-CSS unitaries,
$\Uh$ is a tensor product of phase-adjusted Hadamard gates and identity gates, and $P$ is a multi-qubit Pauli operator. 

Unitaries $\Uph, \Vph$ have eigenstate $|\textbf{0}\rangle$, which allows us to conveniently fix their global phase, as it is done in CH decomposition which represents stabilizer states as $\Uph \Uh |\textbf{0}\rangle$. 
For an efficient phase-sensitive representation of a Clifford unitary we 
keep track of $\Uph, \Vph$ using the standard binary-symplectic representation. Phases of $\Uh$ and $P$ 
are easy to fix too because they are both tensor products of $2 \times 2$ matrices. 
We refer to representation $\ph^m \Uph \Uh P \Vph$ for $\ph = e^{i\frac{\pi}{4}}$ as a phased Bruhat decomposition.

For the purposes of our simulation algorithm, we next show how to efficiently update Bruhat decomposition of phased Clifford unitaries when left-multiplying by other Clifford unitaries~(\cref{alg:phased-left-mul}). 
Left multiplication by Pauli exponents $e^{i\frac{\pi}{4}P}$ is of particular interest for two reasons. 
First, it is used for simulating measurements with uniformly random outcomes.
Second, any Clifford unitary can be written as a product of Pauli exponents~\cite{PllahaVolantoTirkkonen,OMearaSymplectic},
with a number of exponents linear in the dimension of the unitary.

\begin{figure*}[h]
\begin{algorithm}[\texttt{Phase sensitive left-multiplication by Pauli exponent}] \label{alg:phased-left-mul}
\begin{algorithmic}[1]
\Blank
\Input A Clifford unitary with a phase $C = \ph^{m} \Uph \Uh P \Vph$, Pauli exponent $e^{i\frac{\pi}{4}Q}$,
where $\Uph, \Vph$ are phase-CSS, $\Uh$ is a tensor product of $\tilde H$, $I$, and $P$ is a Pauli operator.
\Output 
$m, \Uph, \Uh, P, \Vph$ has been updated so that $\ph^{m} \Uph \Uh P \Vph = e^{i\frac{\pi}{4}Q} C$

\If{ $Q$ is a tensor product of $Z$ and $I$ }
     \State replace $\Uph \leftarrow e^{i\frac{\pi}{4}Q} \Uph$,
\Else{}
    \State let $\tilde Q \leftarrow (\Uph)^{\dagger} Q \Uph$ \Comment $e^{i\frac{\pi}{4}Q}\Uph \Uh P \Vph = \Uph e^{i\frac{\pi}{4}\tilde Q} \Uh P \Vph $
    \State let $J_H$ be indices of $\tilde H$ in $\Uh$, let $J_I = [n] - J_H$ be the complement of $J_H$
    \State let $\tilde Q_I, \tilde Q_H$ factorize $\tilde Q = \tilde Q_I \tilde Q_H$, support of $\tilde Q_I$ is $J_I$, support of $\tilde Q_H$ is $J_H$
    \State let $\mathcal{C}_I, \mathcal{C}_H$ be CSS circuits that map $\tilde Q_I, \tilde Q_H$ to $Q'_I, Q'_H$ in $\{ \pm I,\pm X, \pm Y, \pm Z, \pm YY \}$
    \Blank  \Comment{ \texttt{CSS-Orbit} \cref{alg:css-orbit}}
    \State let $Q' = Q'_H Q'_I$, let $\mathcal{C}^\#_H$ be the Hadamard-conjugated $\mathcal{C}_H$
    \State replace $\Uph \leftarrow \Uph \, \mathcal{C}^{-1}_I \, \mathcal{C}^{-1}_H$, $P \leftarrow \mathcal{C}_I \, \mathcal{C}^\#_H \, P \, (\mathcal{C}_I \, \mathcal{C}^\#_H )^{-1}$  ,$\Vph \leftarrow \mathcal{C}_I \, \mathcal{C}^\#_H \, \Vph$
    \State let $U'$ be $\Uh$ restricted to support of $Q'$ 
    \State lookup decomposition $\ph^{m'} \Uph' \Uh' R' \Vph'$ for $e^{i\frac{\pi}{4}Q} U'$ 
    \Blank \Comment{there are finitely many options for $e^{i\frac{\pi}{4}Q} U'$ listed in \cref{tab:bruhat-decompositions}}
    \State replace $m \leftarrow m + m'$,  $\Uh \leftarrow \Uh \Uh'$, $\Uph \leftarrow \Uph \Uph'$, $\Vph \leftarrow \Vph' \Vph$, $P \leftarrow P' ( (\Vph')^\dagger P \Vph')$
\EndIf
\end{algorithmic}
\end{algorithm}
\end{figure*}

We conclude this subsection with a sketch of the correctness proof of~\cref{alg:phased-left-mul}.
The algorithm proceeds in two main steps. 

First, we use the equality 
$$
e^{i\frac{\pi}{4}Q}\Uph \Uh P \Vph = \Uph e^{i\frac{\pi}{4}\tilde Q} \Uh P \Vph
$$
to reduce the question to computing phased Bruhat decomposition of $e^{i\frac{\pi}{4}\tilde Q} \Uh P$.
Indeed, it is easy to calculate phased Bruhat decomposition of $\Uph C \Vph$ if the phased Bruhat decomposition 
of $C$ is known.

Second, we reduce computing the phased Bruhat decomposition of $e^{i\frac{\pi}{4}\tilde Q} \Uh P$ 
to phased Bruhat decompositions of certain four-qubit phased Clifford unitaries. 
This is best illustrated by an example where $\Uh = \tilde H^{\otimes n} \otimes I^{\otimes m}$ and $P = I$.
Let $U_A, U_B$ be CSS Clifford unitaries on $n$ and $m$ qubits. 
It is easy to compute the phased Bruhat decomposition of $e^{i\frac{\pi}{4}Q} \Uh$, if we know the phased Bruhat decomposition of 
\begin{align*}
 (U_A \otimes U_B) e^{i\frac{\pi}{4}Q} \Uh (U_A \otimes U_B)^\dagger = & (U_A \otimes U_B) e^{i\frac{\pi}{4}Q} (U_A \otimes U_B)^\dagger  (U_A \otimes U_B) \Uh (U_A \otimes U_B)^\dagger = \\
   = & (U_A \otimes U_B) e^{i\frac{\pi}{4}Q} (U_A \otimes U_B)^\dagger \Uh U_{A^{-T}A}
\end{align*}
By appropriately choosing CSS Clifford unitaries $U_A,U_B$, we can ensure that the weight of $Q'$ in 
$$
e^{i\frac{\pi}{4}Q'} = (U_A \otimes U_B) e^{i\frac{\pi}{4}\tilde Q} (U_A \otimes U_B)^\dagger
$$
is at most four. 
This is because for any Pauli operator $P'$, there is a CSS Clifford $U_A'$ such that the
weight of $U_A' P' U_{A'}^\dagger$ is at most two~(\cref{alg:css-orbit}).
We have reduced the original problem to computing a phased Bruhat decomposition 
of the unitary $e^{i\frac{\pi}{4}Q'} U'$ supported on up to four qubits, where $U'$ is the restriction of $\Uh$ to the support of $Q'$. There are finitely many such unitaries, with their phased Bruhat decompositions provided in~\cref{tab:bruhat-decompositions}.

\subsection{Measurement as unitary}
\label{sec:measurement-as-unitary}

Below is a direct generalization of Proposition~2.2 in~\cite{KBP} that includes the global phase. 

\begin{proposition}[Measurement as unitary with phase] 

\label{prop:measure-as-exp}
Let $|\psi\rangle$ be a state stabilized by a Pauli observable $(-1)^b Q$ for $b \in \{0,1\}$,
and let $P$ be a Pauli observable that anticommutes with $Q$. 
The probability of outcome zero of measuring $P$ is $1/2$. 
For measurement outcome $r$, the resulting state is $(-1)^{b(r+b)} Q^{r+b} e^{i\frac{\pi}{4}(iQP)}|\psi\rangle$.
\end{proposition}
\begin{proof}
The probability of outcome zero is equal to the probability of outcome one because: 
$$
\langle \psi | \frac{I+P}{2} | \psi \rangle
=
\langle \psi | (-1)^b Q\frac{I+P}{2} (-1)^b Q | \psi \rangle
=
\langle \psi | \frac{I-P}{2} | \psi \rangle.
$$
For outcome $r$, the state is equal to $\frac{(I+(-1)^r P)}{\sqrt{2}} |\psi\rangle$.
To complete the proof we check that $ Q^{r+b} e^{i\frac{\pi}{4}(iQP)}|\psi\rangle = (-1)^{b(r+b)} \cdot \frac{(I+(-1)^r P)}{\sqrt{2}} |\psi\rangle$ by inspection:
\begin{align*}
Q^{r+b} e^{i\frac{\pi}{4}(iQP)}|\psi\rangle & = \frac{Q^{r+b}}{\sqrt{2}}|\psi\rangle - \frac{Q^{r+b} Q P}{\sqrt{2}}|\psi\rangle \\
 & = \frac{(-1)^{b(r+b)}}{\sqrt{2}}|\psi\rangle + \frac{ (-1)^{b+r + b(r+b+1)} P }{\sqrt{2}}|\psi\rangle, \\
 & = (-1)^{b(r+b)} \cdot \frac{(I+(-1)^r P)}{\sqrt{2}} |\psi\rangle,
\end{align*}
where we used that $Q^{r+b}\ket{\psi}=(-1)^{b(r+b)}\ket{\psi}$, $Q^{r+b +1}P = (-1)^{r+b+1}P Q^{r+b+1}$ and $e^{i \frac{\pi}{4} \phi P'} =
 I \cos(\phi) + i \sin(\phi) P'$.
\end{proof}

\subsection{Separating auxiliary qubits}
\label{sec:aux-separation}

We assume that auxiliary qubits are disentangled at the end of simulation.
To separate the auxiliary qubits, we follow Appendix~C in \cite{KBP}
and apply the bipartite normal form for families of stabilizer states.
At the end of phased outcome-complete simulation, the joint state of output and auxiliary qubits is (see \cref{def:outcome-complete-phase-stab-sim})
$$
i^{\ip{p,r}} (-1)^{\ip{Br+s,r}} \Co | A r \rangle 
$$
Next we show  how to find Clifford unitaries $\Co_1, \Co_2$ (supported on output and auxiliary qubits correspondingly), matrices $A_1, A_2$ and phase $i^l$ such that 
\begin{equation} \label{eq:aux-out-separation}
i^{\ip{p,r}} (-1)^{\ip{Br+s,r}} \Co | A r \rangle  = i^l i^{\ip{p,r}} (-1)^{\ip{Br+s,r}} (\Co_1 |A_1 r\rangle \otimes \Co_2 | A_2 r \rangle )      
\end{equation}

It follows from the bipartite normal form for a family of stabilizer states (see Problem~6.1 in~\cite{KBP}) that there exists a CSS Clifford unitary $U_F$ such that $\Co U_F \simeq C_1 \otimes C_2$, where $C_1$ is supported on 
the output qubits and $C_2$ is supported on the auxiliary qubits. This holds because the auxiliary qubits are disentangled, and the number of Bell pairs involved in the bipartite form is zero.
To recover the global phase in the up-to-phase equality  $\Co U_F \simeq C_1 \otimes C_2$, we decompose $C_1$ and 
$C_2$ as products $\Pi_1, \Pi_2$ of Pauli exponents $e^{i\frac{\pi}{4}P}$. 
Using \cref{alg:phased-left-mul} for multiplying a phased Clifford by Pauli exponents, we compute 
$$
i^l I = \Pi_1^\dagger \Pi_2^\dagger \Co U_F 
$$
which is identity up to global phase. Similarly, using \cref{alg:phased-left-mul} we compute $\Co_1, \Co_2$ from products $\Pi_1, \Pi_2$, and observe $A' = F A$ (using Proposition A.2 in \cite{KBP}).
Finally, using notation $(A_1 r) \oplus (A_2 r) = A'r$ with dimensions of $A_1 r$ and $A_2 r$ equal to qubit counts of $C_1, C_2$ we get \cref{eq:aux-out-separation}.

\section{Outlook}

We have introduced an outcome-complete simulation algorithm that tracks global phases. 
Beyond its application to the verification of non-stabilizer circuits described in \cref{sec:verification},
the new algorithm may be useful for analyzing fault-tolerant circuits in the presence of loss.
For example, certain kinds of qubit loss can be modeled by inserting projectors on $(I+Z)/2$
into the circuit. The channel implemented by such a circuit can be computed as a linear combination 
of two channels with $Z$ and $I$ in the location of the projector. 
Another potential application of our algorithm is speeding up low T-count simulation algorithms~\cite{BravyiBrowneAndCo}, especially for circuits that include many measurements with conditional Pauli corrections. Several questions related to the computational efficiency of our algorithm may improve its practicality:
\begin{enumerate}
    \item What is the most efficient way to represent matrix $B$ in \cref{alg:outcome-complete-stab-phase-sim}? 
    \item What is the most efficient data structure for phased Clifford unitaries?
    \item Can the scaling of runtime with the number of random bits be improved in both conventional and phased outcome-complete simulation algorithms? 
\end{enumerate}

\section*{Acknowledgments}
V.K. has used Claude Opus 4.5 model in GitHub copilot for proofreading, checking consistency of notation, and drafting of abstract and introduction. This work was completed while V.K. was a researcher at Microsoft Quantum.

\bibliographystyle{plainurl}
\bibliography{references}

\newpage 

\appendix

\section{Orbits of Pauli operators with respect to CSS Clifford unitaries.}
\label{app:css-orbits}

The algorithm below shows that any Pauli operator $P$ can be transformed into a Pauli operator of weight at most two 
by using at most two CSS controlled-Pauli operators, that is, operators $\Lambda(X^x,Z^z)$ with $\ip{x,z} = 0$. 

\begin{figure*}[h]
\begin{algorithm}[\texttt{CSS Orbit}] \label{alg:css-orbit}
\begin{algorithmic}[1]
\Blank
\Input Pauli operator $P$ on $n$ qubits
\Output CSS circuit $\mathcal{C}$ of length at most $2$ and Pauli operator $P'$ of weight at most two, 
such that $\mathcal{C}$ maps $P$ to $P' \in \{\pm I,\pm X_j, \pm Y_j, \pm Z_j, \pm Y_j Y_k : j,k \in [n]\}$.
\State let $l,x,z$  be such that $ i^l X^x Z^z \leftarrow P$
\If{ $x = 0,\,z = 0$ }
     \State \Return Empty circuit $\mathcal{C}$, $P' = P$
\ElsIf{$x = 0,\,z \ne 0$}
    \State let $j$ be a non-zero index of $z$,
    \Return $\mathcal{C} = \Lambda(X_j,Z^{z+ e_j})$, $P' = i^l Z_j$ \label{line:css-orbit-z}
\ElsIf{$x \ne 0,\,z = 0$}
    \State let $j$ be a non-zero index of $x$, 
    \Return $\mathcal{C} = \Lambda(Z_j, X^{x+ e_j})$, $P' = i^l X_j$ \label{line:css-orbit-x}
\ElsIf{$x \ne 0,\, z \ne 0, \ip{z,x} = 1$}
    \State let $j$ be a non-zero index of both $x$ and $z$
    \State \Return $\mathcal{C} = \Lambda(Z_j, X^{x + e_j}) \circ \Lambda(X_j, Z^{z + e_j})$, $P' = i^l X_j Z_j$  \label{line:css-orbit-y}
\Else{ $x \ne 0,\, z\ne 0, \ip{z,x} = 0$}
    \State let $j,k$ be non-zero indices of both $x$ and $z$
    \State \Return $\mathcal{C} = \Lambda(Z_j, X^{x + e_j}) \circ \Lambda(X_k, Z^{z + e_k})$, $P' = -i^l Y_j Y_k$  \label{line:css-orbit-xz}
\EndIf
\end{algorithmic}
\end{algorithm}
\end{figure*}

\begin{proposition}
\cref{alg:css-orbit} is correct.
\end{proposition}
\begin{proof}
We check that circuit $\mathcal{C}$ maps the input Pauli $P$ by inspection using the following 
observation, which is a direct consequence of \cref{eq:controlled-pauli-image}: 
$$
 \Lambda(X^a,Z^b)\, X^x \, \Lambda(X^a,Z^b)^\dagger = X^{x + a\cdot\ip{b,x} },\quad \Lambda(X^a,Z^b)\,Z^z\, \Lambda(X^a,Z^b)^\dagger = Z^{z + b\cdot\ip{a,z} }  
$$
When $P$ is equal to $X^x$ or $Z^z$ up to a phase (\cref{line:css-orbit-x,line:css-orbit-z} in \cref{alg:css-orbit}), we get 
$$
 x + (x+e_j) \ip{e_j,x} = e_j, \text{ or } z + (z+e_j) \ip{e_j,z} = e_j,
$$
for $x$ or $z$ bits correspondingly.
When $P$ is equal to $i^l X^x Z^z$ with $\ip{x,z} = 1$, $x,z \ne 0$ (\cref{line:css-orbit-y} in \cref{alg:css-orbit}) after applying the first generalized controlled Pauli the $x,z$ bits become:
$$
 x + (x+e_j) \ip{e_j,x} = e_j, \quad z + e_j \ip{x + e_j,z} = z ,
$$
and after the second generalized controlled Pauli the $x,z$ bits become:
$$
e_j + e_j \ip{z+e_j,e_j} = e_j, \quad z + ( z + e_j) \ip{e_j,z} = e_j.
$$
When $P$ is equal to $i^l X^x Z^z$ with $\ip{x,z} = 0$, $x,z \ne 0$ (\cref{line:css-orbit-xz} in \cref{alg:css-orbit}) after applying the first generalized controlled Pauli the $x,z$ bits become:
$$
 x + (x+e_j) \ip{e_j,x} = e_j, \quad z + e_j \ip{x + e_j,z} = z + e_j,
$$
and after the second generalized controlled Pauli the $x,z$ bits become:
$$
e_j + e_k\ip{e_j,z+e_k} = e_j + e_k, \quad (z+e_j) + (z+e_k) \ip{e_k,z+e_j} = e_j + e_k.
$$
\end{proof}

\newpage 

\section{Explicit Bruhat decompositions of Pauli exponents followed by Hadamard gates} 
\label{tab:bruhat-decompositions}
{
\small
Below we provide Bruhat decompositions of Pauli exponents followed by Hadamard gates. We use $\tilde h = e^{\frac{i\pi}{4}} H$ to avoid $e^{\frac{i \pi}{4}}$ in the phase expression, lowercase letters for common gates, and ${(\uparrow\downarrow)}_{i,j}$ for the Swap gate to make the table more compact. A table like this can be easily computed using a breadth-first search algorithm similar to~\cite{KliuchnikovMaslov}.

$\,$
}

    {
    \small
    \begin{tabular}{||c||c|l|l|l|l||}
    
    $e^{-\frac{i\pi}{4}} i^l \Uph \Uh P \Vph$ & $i^l$ & $\Uph$ & $\Uh$ & P & $\Vph$ \\ \hline
$e^{\frac{i\pi}{4} X_2}\mathrm{\tilde{h}}_{2}$ & $-i$ & 
   $\mathrm{s}_{2}$ & $\mathrm{\tilde{h}}_{2}$ & $\mathrm{x}_{2}$ & $\,$ \\
$e^{\frac{i\pi}{4} Y_2}\mathrm{\tilde{h}}_{2}$ & $-i$ & 
   $\,$ & $\,$ & $\mathrm{x}_{2}$ & $\,$ \\
$e^{\frac{i\pi}{4} Z_2}\mathrm{\tilde{h}}_{2}$ & $-i$ & 
   $\,$ & $\mathrm{\tilde{h}}_{2}$ & $\mathrm{z}_{2}$ & $\mathrm{s}_{2}$ \\
$e^{\frac{i\pi}{4} Y_2 Y_3}\mathrm{\tilde{h}}_{2} \mathrm{\tilde{h}}_{3}$ & $1$ & 
   $\mathrm{s}_{2} {(\uparrow\downarrow)}_{3,2} \mathrm{c\mathrm{x}}_{2,3}$ & $\mathrm{\tilde{h}}_{2} \mathrm{\tilde{h}}_{3}$ & $\,$ & $\mathrm{c\mathrm{x}}_{2,3} \mathrm{s}_{3}$ \\
$e^{\frac{i\pi}{4}X_0}\,$ & $1$ & 
   $\mathrm{s}_{0}$ & $\mathrm{\tilde{h}}_{0}$ & $\,$ & $\mathrm{s}_{0}$ \\
$e^{\frac{i\pi}{4}X_0 X_2}\mathrm{\tilde{h}}_{2}$ & $1$ & 
   $\mathrm{s}_{0} \mathrm{c\mathrm{x}}_{2,0} \mathrm{s}_{2}$ & $\mathrm{\tilde{h}}_{0} \mathrm{\tilde{h}}_{2}$ & $\mathrm{x}_{2}$ & $\mathrm{s}_{0}$ \\
$e^{\frac{i\pi}{4}X_0 Y_2}\mathrm{\tilde{h}}_{2}$ & $1$ & 
   $\mathrm{c\mathrm{z}}_{2,0} {(\uparrow\downarrow)}_{2,0} \mathrm{c\mathrm{x}}_{2,0}$ & $\mathrm{\tilde{h}}_{0} \mathrm{\tilde{h}}_{2}$ & $\mathrm{z}_{0} \mathrm{x}_{2}$ & $\mathrm{c\mathrm{z}}_{2,0}$ \\
$e^{\frac{i\pi}{4}X_0 Z_2}\mathrm{\tilde{h}}_{2}$ & $1$ & 
   $\mathrm{s}_{0}$ & $\mathrm{\tilde{h}}_{0} \mathrm{\tilde{h}}_{2}$ & $\mathrm{z}_{2}$ & $\mathrm{c\mathrm{x}}_{2,0} \mathrm{s}_{2} \mathrm{s}_{0}$ \\
$e^{\frac{i\pi}{4}X_0 Y_2 Y_3}\mathrm{\tilde{h}}_{2} \mathrm{\tilde{h}}_{3}$ & $i$ & 
   $\mathrm{c\mathrm{z}}_{3,0} \mathrm{c\mathrm{x}}_{2,3} \mathrm{s}_{0}$ & $\mathrm{\tilde{h}}_{0} \mathrm{\tilde{h}}_{2} \mathrm{\tilde{h}}_{3}$ & $\mathrm{y}_{0}$ & $\mathrm{c\mathrm{x}}_{3,2} \mathrm{c\mathrm{x}}_{0,3} \mathrm{c\mathrm{x}}_{2,0} \mathrm{s}_{2} \mathrm{s}_{0}$ \\
$e^{\frac{i\pi}{4}Y_0}\,$ & $1$ & 
   $\,$ & $\mathrm{\tilde{h}}_{0}$ & $\mathrm{x}_{0}$ & $\,$ \\
$e^{\frac{i\pi}{4}Y_0 X_2}\mathrm{\tilde{h}}_{2}$ & $1$ & 
   $\mathrm{c\mathrm{z}}_{2,0} \mathrm{c\mathrm{x}}_{2,0}$ & $\mathrm{\tilde{h}}_{0} \mathrm{\tilde{h}}_{2}$ & $\mathrm{x}_{0}$ & $\,$ \\
$e^{\frac{i\pi}{4}Y_0 Y_2}\mathrm{\tilde{h}}_{2}$ & $i$ & 
   $\mathrm{c\mathrm{z}}_{2,0} {(\uparrow\downarrow)}_{2,0} \mathrm{s}_{0}$ & $\mathrm{\tilde{h}}_{0} \mathrm{\tilde{h}}_{2}$ & $\mathrm{y}_{2}$ & $\mathrm{c\mathrm{x}}_{2,0} \mathrm{s}_{2} \mathrm{s}_{0}$ \\
$e^{\frac{i\pi}{4}Y_0 Z_2}\mathrm{\tilde{h}}_{2}$ & $1$ & 
   $\mathrm{c\mathrm{x}}_{0,2}$ & $\mathrm{\tilde{h}}_{0} \mathrm{\tilde{h}}_{2}$ & $\mathrm{x}_{0}$ & $\mathrm{c\mathrm{z}}_{2,0}$ \\
$e^{\frac{i\pi}{4}Y_0 Y_2 Y_3}\mathrm{\tilde{h}}_{2} \mathrm{\tilde{h}}_{3}$ & $1$ & 
   $\mathrm{c\mathrm{z}}_{3,0} \mathrm{c\mathrm{x}}_{2,3}$ & $\mathrm{\tilde{h}}_{0} \mathrm{\tilde{h}}_{2} \mathrm{\tilde{h}}_{3}$ & $\mathrm{z}_{0} \mathrm{z}_{2}$ & $\mathrm{c\mathrm{x}}_{3,2} \mathrm{c\mathrm{x}}_{0,3} \mathrm{c\mathrm{z}}_{2,0} \mathrm{c\mathrm{x}}_{2,0}$ \\
$e^{\frac{i\pi}{4}Z_0}\,$ & $-i$ & 
   $\,$ & $\,$ & $\mathrm{z}_{0}$ & $\mathrm{s}_{0}$ \\
$e^{\frac{i\pi}{4}Z_0 X_2}\mathrm{\tilde{h}}_{2}$ & $-i$ & 
   $\mathrm{c\mathrm{z}}_{2,0} \mathrm{s}_{2} \mathrm{s}_{0}$ & $\mathrm{\tilde{h}}_{2}$ & $\mathrm{z}_{0} \mathrm{x}_{2}$ & $\,$ \\
$e^{\frac{i\pi}{4}Z_0 Y_2}\mathrm{\tilde{h}}_{2}$ & $-i$ & 
   $\,$ & $\,$ & $\mathrm{x}_{2}$ & $\mathrm{c\mathrm{z}}_{2,0} \mathrm{c\mathrm{x}}_{0,2}$ \\
$e^{\frac{i\pi}{4}Z_0 Z_2}\mathrm{\tilde{h}}_{2}$ & $-i$ & 
   $\mathrm{c\mathrm{x}}_{0,2} \mathrm{s}_{0}$ & $\mathrm{\tilde{h}}_{2}$ & $\mathrm{z}_{0} \mathrm{z}_{2}$ & $\mathrm{s}_{2}$ \\
$e^{\frac{i\pi}{4}Z_0 Y_2 Y_3}\mathrm{\tilde{h}}_{2} \mathrm{\tilde{h}}_{3}$ & $1$ & 
   $\mathrm{c\mathrm{x}}_{0,2} \mathrm{s}_{2} {(\uparrow\downarrow)}_{3,2} \mathrm{c\mathrm{x}}_{2,3} \mathrm{c\mathrm{x}}_{0,2}$ & $\mathrm{\tilde{h}}_{2} \mathrm{\tilde{h}}_{3}$ & $\,$ & $\mathrm{c\mathrm{x}}_{2,3} \mathrm{s}_{3}$ \\
$e^{\frac{i\pi}{4}Y_0 Y_1}\,$ & $i$ & 
   $\mathrm{c\mathrm{x}}_{1,0} \mathrm{s}_{0}$ & $\mathrm{\tilde{h}}_{1}$ & $\mathrm{z}_{0} \mathrm{y}_{1}$ & $\mathrm{s}_{0} \mathrm{c\mathrm{x}}_{1,0}$ \\
$e^{\frac{i\pi}{4}Y_0 Y_1 X_2}\mathrm{\tilde{h}}_{2}$ & $i$ & 
   $\mathrm{c\mathrm{z}}_{2,1} \mathrm{c\mathrm{x}}_{2,1} \mathrm{c\mathrm{x}}_{1,0} \mathrm{s}_{0}$ & $\mathrm{\tilde{h}}_{1} \mathrm{\tilde{h}}_{2}$ & $\mathrm{z}_{0} \mathrm{y}_{1}$ & $\mathrm{s}_{0} \mathrm{c\mathrm{x}}_{1,0}$ \\
$e^{\frac{i\pi}{4}Y_0 Y_1 Y_2}\mathrm{\tilde{h}}_{2}$ & $1$ & 
   $\mathrm{c\mathrm{z}}_{2,1} \mathrm{c\mathrm{x}}_{2,1} \mathrm{c\mathrm{x}}_{1,0} \mathrm{c\mathrm{x}}_{0,2}$ & $\mathrm{\tilde{h}}_{1} \mathrm{\tilde{h}}_{2}$ & $\mathrm{x}_{1}$ & $\mathrm{c\mathrm{z}}_{2,0} \mathrm{c\mathrm{x}}_{1,0}$ \\
$e^{\frac{i\pi}{4}Y_0 Y_1 Z_2}\mathrm{\tilde{h}}_{2}$ & $i$ & 
   $\mathrm{c\mathrm{x}}_{1,0} \mathrm{c\mathrm{x}}_{0,2} \mathrm{s}_{0}$ & $\mathrm{\tilde{h}}_{1} \mathrm{\tilde{h}}_{2}$ & $\mathrm{z}_{0} \mathrm{y}_{1}$ & $\mathrm{s}_{0} \mathrm{c\mathrm{z}}_{2,0} \mathrm{c\mathrm{x}}_{1,0}$ \\
$e^{\frac{i\pi}{4}Y_0 Y_1 Y_2 Y_3}\mathrm{\tilde{h}}_{2} \mathrm{\tilde{h}}_{3}$ & $1$ & 
   $\mathrm{c\mathrm{x}}_{1,0} \mathrm{c\mathrm{x}}_{0,3} \mathrm{c\mathrm{z}}_{2,0} \mathrm{c\mathrm{x}}_{3,2} \mathrm{s}_{0}$ & $\mathrm{\tilde{h}}_{1} \mathrm{\tilde{h}}_{2} \mathrm{\tilde{h}}_{3}$ & $\mathrm{z}_{0}$ & $\mathrm{c\mathrm{x}}_{2,3} \mathrm{c\mathrm{x}}_{0,2} \mathrm{s}_{0} \mathrm{c\mathrm{z}}_{3,0} \mathrm{c\mathrm{x}}_{1,0}$ \\
$e^{-\frac{i\pi}{4}X_0}\,$ & $i$ & 
   $\mathrm{s}_{0}$ & $\mathrm{\tilde{h}}_{0}$ & $\mathrm{y}_{0}$ & $\mathrm{s}_{0}$ \\
$e^{-\frac{i\pi}{4}X_0 X_2}\mathrm{\tilde{h}}_{2}$ & $i$ & 
   $\mathrm{s}_{0} \mathrm{c\mathrm{x}}_{2,0} \mathrm{s}_{2}$ & $\mathrm{\tilde{h}}_{0} \mathrm{\tilde{h}}_{2}$ & $\mathrm{y}_{0}$ & $\mathrm{s}_{0}$ \\
$e^{-\frac{i\pi}{4}X_0 Y_2}\mathrm{\tilde{h}}_{2}$ & $1$ & 
   $\mathrm{c\mathrm{z}}_{2,0} {(\uparrow\downarrow)}_{2,0} \mathrm{c\mathrm{x}}_{2,0}$ & $\mathrm{\tilde{h}}_{0} \mathrm{\tilde{h}}_{2}$ & $\mathrm{x}_{0}$ & $\mathrm{c\mathrm{z}}_{2,0}$ \\
$e^{-\frac{i\pi}{4}X_0 Z_2}\mathrm{\tilde{h}}_{2}$ & $i$ & 
   $\mathrm{s}_{0}$ & $\mathrm{\tilde{h}}_{0} \mathrm{\tilde{h}}_{2}$ & $\mathrm{y}_{0}$ & $\mathrm{c\mathrm{x}}_{2,0} \mathrm{s}_{2} \mathrm{s}_{0}$ \\
$e^{-\frac{i\pi}{4}X_0 Y_2 Y_3}\mathrm{\tilde{h}}_{2} \mathrm{\tilde{h}}_{3}$ & $1$ & 
   $\mathrm{c\mathrm{z}}_{3,0} \mathrm{c\mathrm{x}}_{2,3} \mathrm{s}_{0}$ & $\mathrm{\tilde{h}}_{0} \mathrm{\tilde{h}}_{2} \mathrm{\tilde{h}}_{3}$ & $\mathrm{z}_{2}$ & $\mathrm{c\mathrm{x}}_{3,2} \mathrm{c\mathrm{x}}_{0,3} \mathrm{c\mathrm{x}}_{2,0} \mathrm{s}_{2} \mathrm{s}_{0}$ \\
$e^{-\frac{i\pi}{4}Y_0}\,$ & $1$ & 
   $\,$ & $\mathrm{\tilde{h}}_{0}$ & $\mathrm{z}_{0}$ & $\,$ \\
$e^{-\frac{i\pi}{4}Y_0 X_2}\mathrm{\tilde{h}}_{2}$ & $1$ & 
   $\mathrm{c\mathrm{z}}_{2,0} \mathrm{c\mathrm{x}}_{2,0}$ & $\mathrm{\tilde{h}}_{0} \mathrm{\tilde{h}}_{2}$ & $\mathrm{z}_{0} \mathrm{x}_{2}$ & $\,$ \\
$e^{-\frac{i\pi}{4}Y_0 Y_2}\mathrm{\tilde{h}}_{2}$ & $1$ & 
   $\mathrm{c\mathrm{z}}_{2,0} {(\uparrow\downarrow)}_{2,0} \mathrm{s}_{0}$ & $\mathrm{\tilde{h}}_{0} \mathrm{\tilde{h}}_{2}$ & $\mathrm{z}_{0}$ & $\mathrm{c\mathrm{x}}_{2,0} \mathrm{s}_{2} \mathrm{s}_{0}$ \\
$e^{-\frac{i\pi}{4}Y_0 Z_2}\mathrm{\tilde{h}}_{2}$ & $1$ & 
   $\mathrm{c\mathrm{x}}_{0,2}$ & $\mathrm{\tilde{h}}_{0} \mathrm{\tilde{h}}_{2}$ & $\mathrm{z}_{0}$ & $\mathrm{c\mathrm{z}}_{2,0}$ \\
$e^{-\frac{i\pi}{4}Y_0 Y_2 Y_3}\mathrm{\tilde{h}}_{2} \mathrm{\tilde{h}}_{3}$ & $1$ & 
   $\mathrm{c\mathrm{z}}_{3,0} \mathrm{c\mathrm{x}}_{2,3}$ & $\mathrm{\tilde{h}}_{0} \mathrm{\tilde{h}}_{2} \mathrm{\tilde{h}}_{3}$ & $\mathrm{x}_{0}$ & $\mathrm{c\mathrm{x}}_{3,2} \mathrm{c\mathrm{x}}_{0,3} \mathrm{c\mathrm{z}}_{2,0} \mathrm{c\mathrm{x}}_{2,0}$ \\
$e^{-\frac{i\pi}{4}Z_0}\,$ & $1$ & 
   $\,$ & $\,$ & $\,$ & $\mathrm{s}_{0}$ \\
$e^{-\frac{i\pi}{4}Z_0 X_2}\mathrm{\tilde{h}}_{2}$ & $1$ & 
   $\mathrm{c\mathrm{z}}_{2,0} \mathrm{s}_{2} \mathrm{s}_{0}$ & $\mathrm{\tilde{h}}_{2}$ & $\,$ & $\,$ \\
$e^{-\frac{i\pi}{4}Z_0 Y_2}\mathrm{\tilde{h}}_{2}$ & $-i$ & 
   $\,$ & $\,$ & $\mathrm{z}_{0} \mathrm{z}_{2}$ & $\mathrm{c\mathrm{z}}_{2,0} \mathrm{c\mathrm{x}}_{0,2}$ \\
$e^{-\frac{i\pi}{4}Z_0 Z_2}\mathrm{\tilde{h}}_{2}$ & $1$ & 
   $\mathrm{c\mathrm{x}}_{0,2} \mathrm{s}_{0}$ & $\mathrm{\tilde{h}}_{2}$ & $\,$ & $\mathrm{s}_{2}$ \\
$e^{-\frac{i\pi}{4}Z_0 Y_2 Y_3}\mathrm{\tilde{h}}_{2} \mathrm{\tilde{h}}_{3}$ & $-i$ & 
   $\mathrm{c\mathrm{x}}_{0,2} \mathrm{s}_{2} {(\uparrow\downarrow)}_{3,2} \mathrm{c\mathrm{x}}_{2,3} \mathrm{c\mathrm{x}}_{0,2}$ & $\mathrm{\tilde{h}}_{2} \mathrm{\tilde{h}}_{3}$ & $\mathrm{z}_{0} \mathrm{x}_{2} \mathrm{z}_{3}$ & $\mathrm{c\mathrm{x}}_{2,3} \mathrm{s}_{3}$ \\
$e^{-\frac{i\pi}{4}Y_0 Y_1}\,$ & $1$ & 
   $\mathrm{c\mathrm{x}}_{1,0} \mathrm{s}_{0}$ & $\mathrm{\tilde{h}}_{1}$ & $\mathrm{z}_{0}$ & $\mathrm{s}_{0} \mathrm{c\mathrm{x}}_{1,0}$ \\
$e^{-\frac{i\pi}{4}Y_0 Y_1 X_2}\mathrm{\tilde{h}}_{2}$ & $1$ & 
   $\mathrm{c\mathrm{z}}_{2,1} \mathrm{c\mathrm{x}}_{2,1} \mathrm{c\mathrm{x}}_{1,0} \mathrm{s}_{0}$ & $\mathrm{\tilde{h}}_{1} \mathrm{\tilde{h}}_{2}$ & $\mathrm{z}_{0} \mathrm{x}_{2}$ & $\mathrm{s}_{0} \mathrm{c\mathrm{x}}_{1,0}$ \\
$e^{-\frac{i\pi}{4}Y_0 Y_1 Y_2}\mathrm{\tilde{h}}_{2}$ & $1$ & 
   $\mathrm{c\mathrm{z}}_{2,1} \mathrm{c\mathrm{x}}_{2,1} \mathrm{c\mathrm{x}}_{1,0} \mathrm{c\mathrm{x}}_{0,2}$ & $\mathrm{\tilde{h}}_{1} \mathrm{\tilde{h}}_{2}$ & $\mathrm{z}_{1} \mathrm{x}_{2}$ & $\mathrm{c\mathrm{z}}_{2,0} \mathrm{c\mathrm{x}}_{1,0}$ \\
$e^{-\frac{i\pi}{4}Y_0 Y_1 Z_2}\mathrm{\tilde{h}}_{2}$ & $1$ & 
   $\mathrm{c\mathrm{x}}_{1,0} \mathrm{c\mathrm{x}}_{0,2} \mathrm{s}_{0}$ & $\mathrm{\tilde{h}}_{1} \mathrm{\tilde{h}}_{2}$ & $\mathrm{z}_{0}$ & $\mathrm{s}_{0} \mathrm{c\mathrm{z}}_{2,0} \mathrm{c\mathrm{x}}_{1,0}$ \\
$e^{-\frac{i\pi}{4}Y_0 Y_1 Y_2 Y_3}\mathrm{\tilde{h}}_{2} \mathrm{\tilde{h}}_{3}$ & $i$ & 
   $\mathrm{c\mathrm{x}}_{1,0} \mathrm{c\mathrm{x}}_{0,3} \mathrm{c\mathrm{z}}_{2,0} \mathrm{c\mathrm{x}}_{3,2} \mathrm{s}_{0}$ & $\mathrm{\tilde{h}}_{1} \mathrm{\tilde{h}}_{2} \mathrm{\tilde{h}}_{3}$ & $\mathrm{z}_{0} \mathrm{y}_{1}$ & $\mathrm{c\mathrm{x}}_{2,3} \mathrm{c\mathrm{x}}_{0,2} \mathrm{s}_{0} \mathrm{c\mathrm{z}}_{3,0} \mathrm{c\mathrm{x}}_{1,0}$ 

    \end{tabular}
    }

\end{document}